\documentclass[runningheads]{llncs}
\usepackage[T1]{fontenc}
\usepackage{hyperref}
\usepackage{xcolor}
\usepackage{amsfonts}
\usepackage{amssymb}
\usepackage{amsmath}
\usepackage{todonotes}
\usepackage[capitalize]{cleveref}
\usepackage{soul}

\usepackage{xspace}
\usepackage{twemojis}

\usepackage{tikz}
\usetikzlibrary{automata, arrows.meta, positioning, decorations.pathreplacing}
\usepackage{stmaryrd}

\newcommand{\tup}[1]{\langle #1 \rangle}

\newcommand{\cA}{\mathcal{A}}
\newcommand{\cB}{\mathcal{B}}
\newcommand{\cC}{\mathcal{C}}
\newcommand{\cD}{\mathcal{D}}
\newcommand{\cN}{\mathcal{N}}

\newcommand{\jL}{\mathfrak{J}}
\newcommand{\jfinL}{\mathfrak{J}_{\text{fin}}}
\newcommand{\regL}{\mathfrak{L}}
\newcommand{\kwinL}[1][k]{\mathfrak{J}_{#1\boxplus}}
\newcommand{\ewinL}{\mathfrak{J}_{\exists\boxplus}}
\newcommand{\finL}{\regL_{\text{fin}}}

\newcommand{\bbN}{\mathbb{N}}
\newcommand{\bbM}{\mathbb{M}}

\newcommand{\lin}{\mathrm{Lin}}

\newcommand{\conp}{\texttt{coNP}\xspace}

\newcommand{\nlc}{\texttt{NL-Complete}\xspace}

\newcommand{\DFA}{\mbox{\rm DFA}\xspace}

\newcommand{\Inf}{\mathrm{Inf}}
\newcommand{\parikh}{\Psi}
\newcommand{\numl}[2]{\#_{#2}[#1]}
\newcommand{\simkw}[1][k]{\sim_{k\boxplus}}
\newcommand{\norm}[1]{\|#1\|}

\newcommand{\mask}{\mathfrak{m}}
\newcommand{\maskSetExplicit}{\{0,\infty\}^\Sigma\setminus \{\vec{0}\}}
\newcommand{\maskSet}{\twemoji{performing arts}}

\newcommand{\bbNinf}{\bbM^\Sigma}

\spnewtheorem{complexitya}[theorem]{Complexity Analysis}{\itshape}{\rmfamily}
\crefname{complexitya}{Complexity Analysis}{Complexity Analyses}

\begin{document}
\title{Jumping Automata over Infinite Words\thanks{This research was supported by the ISRAEL SCIENCE FOUNDATION (grant No. 989/22)}}
%
%
\author{Shaull Almagor\inst{1}\orcidID{0000-0001-9021-1175} \and
Omer Yizhaq\inst{2}}
\authorrunning{S. Almagor and O. Yizhaq}
%
\institute{Technion, Israel \email{shaull@technion.ac.il} \and Technion, Israel \email{omeryi@campus.technion.ac.il }}
%
\maketitle              
\begin{abstract}
Jumping automata are finite automata that read their input in a non-consecutive manner, disregarding the order of the letters in the word. We introduce and study jumping automata over infinite words. Unlike the setting of finite words, which has been well studied, for infinite words it is not clear how words can be reordered.
To this end, we consider three semantics: automata that read the infinite word in some order so that no letter is overlooked, automata that can permute the word in windows of a given size k, and automata that can permute the word in windows of an existentially-quantified bound.
We study expressiveness, closure properties and algorithmic properties of these models.

\end{abstract}
\section{Introduction}
\label{sec:intro}
Traditional automata read their input sequentially. Indeed, this is the case for most state-based computational models.
In some settings, however, the order of the input does not matter. For example, when the input represents available resources, and we only wish to reason about their \emph{quantity}. From a more language-theoretic perspective, this amounts to looking at the \emph{commutative closure} of the languages, a.k.a. their \emph{Parikh image}. 
To capture this notion in a computation model, \emph{Jumping Automata} were introduced in~\cite{meduna2012jumping}. A jumping automaton may read its input in a non-sequential manner, jumping from letter to letter, as long as every letter is read exactly once. Several works have studied the algorithmic properties and expressive power of these automata~\cite{fernau2015jumping,fernau2017characterization,vorel2018basic,fazekas2021two,lavado2014operational}. 

One of the most exciting developments in automata and language theory has been the extension to the setting of infinite words~\cite{buchi1966symposium,mcnaughton1966testing,kupferman2018automata}, which has led to powerful tools in formal methods. The infinite-word setting is far more complicated than that of finite words, involving several acceptance conditions and intricate expressiveness relationships. Most notably perhaps, nondeterministic B\"uchi automata cannot be determinized, but can be determinized to Rabin automata~\cite{landweber1968decision,mcnaughton1966testing,safra1988complexity}.

In this work, we introduce jumping automata over infinite words. The first challenge is to find meaningful definitions for this model. Indeed, the intuition of having the reading head ``jump'' to points in the word is ill-suited for infinite words, since one can construct an infinite run even without reading every letter (possibly even skipping infinitely many letters) (see \cref{rmk:jumping_semantics}). 

To this end, we propose three semantics for modes of jumping over infinite words for an automaton $\cA$.
\begin{itemize}
    \item In the \emph{jumping} semantics, a word $w$ is accepted if $\cA$ accepts a permutation of it, i.e., a word $w'$ that has the same number of occurrences of each letter as $w$ for letters occurring finitely often, and the same set of letters that occur infinitely often.  
    \item In the \emph{$k$-window jumping} semantics, $w$ is accepted if we can make $\cA$ accept it by permuting $w$ within contiguous windows of some fixed size $k$.
    \item In the \emph{$\exists$-window jumping} semantics, $w$ is accepted if there exists some $k$ such that we can make $\cA$ accept it by permuting $w$ within contiguous windows of size $k$.
\end{itemize}
\begin{example}
    \label{xmp:semantics}
    Consider a B\"uchi automaton for the language $\{(ab)^\omega\}$ (see~\cref{fig:intersection_NBW_fails}). Its jumping language is $\{w: w$ has infinitely many $a$'s and $b$'s $\}$. Its $3$-window jumping language, for example, consists of words whose $a$'s and $b$'s can be rearranged to construct $(ab)^\omega$ in windows of size $3$, such as $(aab\cdot bba)^\omega$. As for its $\exists$-window jumping language, it contains e.g., the word $(aaaabbbb)^\omega$, which is not in the 3-window language, but does not contain $aba^2b^2a^3b^3\cdots $, which is in the jumping language.
\end{example}
The definitions above capture different intuitive meanings for jumping: the first definition only looks at the Parikh image of the word, and corresponds to e.g., a setting where a word describes resources, with some resources being unbounded. The second (more restrictive) definition captures a setting where the word corresponds to e.g., a sequence of tasks to be handled, and we are allowed some bounded freedom in reordering them, and the third corresponds to a setting where we may introduce any finite delay to a task, but the overall delays we introduce must be bounded.

We study the expressiveness of these semantics, as well as closure properties and decision problems. Specifically, we show that languages in the jumping semantics are closed under union, intersection and complement. Surprisingly, we also show that automata admit determinization in this semantics, in contrast with standard B\"uchi automata. We further show that the complexity of decision problems on these automata coincide with their finite-word jumping counterparts. Technically, these results are attained by augmenting semilinear sets to accommodate $\infty$, and by showing that jumping languages can be described in a canonical form with such sets.

In the $k$-window semantics, we show a correspondence with $\omega$-regular languages, from which we deduce closure properties and solutions for decision problems. Finally, we show that the $\exists$-window semantics is strictly more expressive than the jumping semantics. 


\subsubsection*{Paper organization} In \cref{sec:prelim} we recap some basic notions and define our jumping semantics. \cref{sec:JBW} is the bulk of the paper, where we study the jumping semantics. In \cref{subsec:JBA and semilinear} we introduce our variant of semilinear sets, and prove that it admits a canonical form and that it characterizes the jumping languages. Then, in \cref{subsec:closure_JBA,subsec:algorithms_JBA} we study closure properties and decision problems for jumping languages. In \cref{sec:fixed window,sec:exists_jumping} we study the $k$-window and $\exists$-window semantics, respectively. Finally, we conclude with future research in \cref{sec:future}. 
Throughout the paper, our focus is on clean constructions and decidability. We thus defer complexity analyses to notes that follow each claim. 

\subsubsection*{Related work}
Jumping automata were introduced in~\cite{meduna2012jumping}. We remark that~\cite{meduna2012jumping} contains some erroneous proofs (e.g., closure under intersection and complement, also pointed out in~\cite{fernau2017characterization}). The works in~\cite{fernau2015jumping,fernau2017characterization} establish several expressiveness results on jumping automata, as well as some complexity results. In~\cite{vorel2018basic} many additional closure properties are established. An extension of jumping automata with a two-way tape was studied in~\cite{fazekas2021two}.

Technically, since jumping automata correspond to the semilinear Parikh image/commutative closure~\cite{parikh1966context}, tools on semilinear sets are closely related. The works in~\cite{beierdescriptional,chistikov2016taming,to2010parikh} provide algorithmic results on semilinear sets, which we use extensively, and in~\cite{lavado2014operational} properties of semilinear sets are related to the state complexity of automata.

More broadly, automata and commutative closure also intersect in Parikh Automata~\cite{klaedtke2003monadic,cadilhac2012bounded,cadilhac2012affine,guha2022parikh}, which read their input and accept if a certain Parikh image relating to the run belongs to a given semilinear set. In particular,~\cite{guha2022parikh} studies an extension of these automata to infinite words.
Another related model is that of symmetric transducers -- automata equipped with outputs, such that permutations in the input correspond to permutations in the output. These were studied in~\cite{almagor2020process} in a jumping-flavour, and in~\cite{nassar2022simulation} in a $k$-window flavour. 

\section{Preliminaries and Definitions}
\label{sec:prelim}
\subsubsection{Automata}
A \emph{nondeterministic automaton} is a 5-tuple $\cA=\tup{\Sigma,Q,\delta,Q_0,\alpha}$ where $\Sigma$ is a finite alphabet, $Q$ is a finite set of states, $\delta:Q\times\Sigma\to 2^Q$ is a nondeterministic transition function, $Q_0\subseteq Q$ is a set of initial states, and $\alpha\subseteq Q$ is a set of accepting states. When $|Q_0|=1$ and $|\delta(q,\sigma)|= 1$ for every $q\in Q$ and $\sigma\in \Sigma$, we say that $\cA$ is \emph{deterministic}. 

We consider automata both over finite and over infinite words. We denote by $\Sigma^*$ the set of finite words over $\Sigma$, and by $\Sigma^\omega$ the set of infinite words. For a word $w=\sigma_1,\sigma_2,\ldots$ (either finite or infinite), let $|w|\in \bbN\cup \{\infty\}$ be its length. A \emph{run} of $\cA$ on $w$ is a sequence $\rho=q_0,q_1,\ldots$ such that $q_0\in Q_0$ and for every $0\le i<|w|$ it holds that $q_{i+1}\in \delta(q_i,w_{i+1})$ (naturally, a run is finite if $w$ is finite, and infinite if $w$ is infinite).
For finite words, the run $\rho$ is \emph{accepting} if $q_{|w|}\in \alpha$. For infinite words we use the B\"uchi acceptance condition, whereby the run $\rho$ is accepting if it visits $\alpha$ infinitely often. Formally, define $\Inf(\rho)=\{q\in Q\mid \forall i\in \bbN\ \exists j>i,\ \rho_j=q\}$, then $\rho$ is accepting if $\Inf(\rho)\cap \alpha\neq \emptyset$. A word $w$ is accepted by $\cA$ if there exists an accepting run of $\cA$ on $w$.

The \emph{language} of an automaton $\cA$, denoted $\regL(\cA)$ is the set of words it accepts. We emphasize an automaton is over finite words by writing $\finL(\cA)$.

\begin{remark}
There are other acceptance conditions for automata over infinite words, e.g., parity, Rabin, Streett, co-B\"uchi and Muller. As we show in~\cref{prop:jumping_determinizable}, for our purposes it is enough to consider B\"uchi.
\end{remark}

\subsubsection{Parikh Images and Permutations} 
Fix an alphabet $\Sigma$. We start by defining Parikh images and permutations for finite words.
Consider a finite word $x\in \Sigma^*$. For each letter $\sigma\in \Sigma$ we denote by $\numl{x}{\sigma}$ the number of occurrences of $\sigma$ in $x$. The \emph{Parikh image} of $x$ is then the vector $\parikh(x)=(\numl{x}{\sigma})_{\sigma\in \Sigma}\in \bbN^\Sigma$. We say that a word $y$ is a \emph{permutation} of $x$ if $\parikh(y)=\parikh(x)$, in which case we write $x\sim y$ (clearly $\sim$ is an equivalence relation). We extend the Parikh image to sets of words by $\parikh(L)=\{\parikh(w)\mid w\in L\}$.

We now extend the definitions to infinite words.
Let $\bbN=\{0,1,\ldots\}$ be the non-negative integers, and write $\bbN_\infty=\bbN\cup\{\infty\}$. Consider an infinite word $w\in \Sigma^\omega$. We extend the definition of Parikh image to infinite words in the natural way: $\parikh(w)=(\numl{w}{\sigma})_{\sigma\in \Sigma}$ where $\numl{w}{\sigma}$ is the number of occurrences of $\sigma$ in $w$ if it is finite, and is $\infty$ otherwise. Thus, $\parikh(w)\in \bbN_\infty^{\Sigma}$. Moreover, since $w$ is infinite, at least one coordinate of $\parikh(w)$ is $\infty$, since we often restrict ourselves to the Parikh images of infinite words, we denote by $\bbNinf=\bbN_\infty^\Sigma\setminus \bbN^\Sigma$ the set of vectors that have at least one $\infty$ coordinate. Note that $\vec{v}\in \bbNinf$ if and only if $\vec{v}$ is the Parikh image of some infinite word over $\Sigma$.

For words $w,w'\in \Sigma^\omega$, we abuse notation slightly and write $w\sim w'$ if $\parikh(w)=\parikh(w')$, as well as refer to $w'$ as a \emph{permutation} of $w$. 
We now refine the notion of permutation by restricting permutations to finite ``windows'': let $k\in \bbN$, we say that $w$ is a \emph{$k$-window permutation} of $w'$, and we write $w\simkw w'$ (the symbol $\boxplus$ represents ``window'') if $w=x_1\cdot x_2\cdots$ and $w'=y_1\cdot y_2\cdots$ where for all $i\ge 1$ we have $|x_i|=|y_i|=k$ and $x_i\sim y_i$. That is, $w'$ can be obtained from $w$ by permuting contiguous disjoint windows of size $k$ in $w$. Note that if $w\simkw w'$ then in particular $w\sim w'$, but the former is much more restrictive.

\subsubsection{Semilinear Sets}
Let $\bbN=\{0,1,\ldots\}$ be the natural numbers. For dimension $d\ge 1$, we denote vectors in $\bbN^d$ in bold (e.g., $\vec{p}\in \bbN^d$). Consider a vector $\vec{b}\in \bbN^d$ and $P\subseteq \bbN^d$, the \emph{linear set} generated by the \emph{base} $\vec{b}$ and \emph{periods} $P$ is 
\[
\lin(\vec{b},P)=\{\vec{b}+\sum_{\vec{p}\in P}\lambda_p \vec{p}\mid \vec{b}\in \bbN^d \text{ and }\lambda_p\in \bbN\text{ for all }\vec{p}\in P\}
\]
A \emph{semilinear set} is then $\bigcup_{i\in I} \lin(\vec{b_i},P_i)$ for a finite set $I$ and pairs $(\vec{b_i},P_i)$. 

Semilinear sets are closely related to the Parikh image of regular languages (in that the Parikh image of a regular language is semilinear)~\cite{parikh1966context}. We will cite specific results relating to this connection as we use them.

\subsubsection{Jumping Automata}
Consider an automaton $\cA=\tup{\Sigma,Q,\delta,Q_0,\alpha}$. Over finite words, we view $\cA$ as a \emph{jumping automaton} by letting it read its input in a non-sequential way, i.e. ``jump'' between letters as long as all letters are read exactly once. Formally, $\cA$ accepts a word $w$ as a jumping automaton if it accepts some permutation of $w$. Thus, we define the \emph{jumping language} of $\cA$ to be \[\jfinL(\cA)=\{w\in \Sigma^* \mid \exists w'\sim w\text{ such that } w'\in \finL(\cA)\}\]

We now turn to define jumping automata over infinite words. 
As discussed in \cref{sec:intro}, we consider three jumping semantics for $\cA$, as follows:
\begin{definition}
\label{def:jumping}
Consider an automaton $\cA=\tup{\Sigma,Q,\delta,Q_0,\alpha}$. 
\begin{itemize}
    \item Its \emph{Jumping language} is:
    \\ \hspace*{1cm} $\jL(\cA)=\{w\in \Sigma^\omega\mid \exists w'\sim w\text{ such that } w'\in \regL(\cA)\}$.
\item For $k\in \bbN$, its \emph{$k$-window Jumping language} is:
\\ \hspace*{1cm} $\kwinL(\cA)=\{w\in \Sigma^\omega\mid \exists w'\simkw w\text{ such that } w'\in \regL(\cA)\}$.
\item Its \emph{$\exists$-window Jumping language} is:
\\ \hspace*{1cm} $\ewinL(\cA)=\{w\in \Sigma^\omega\mid \exists k\in \bbN \wedge \exists w'\simkw w\text{ such that } w'\in \regL(\cA)\}$.
\end{itemize}
\end{definition}

\begin{remark}
    \label{rmk:jumping_semantics}
    For jumping automata over finite words, the model is sometimes defined by allowing the transition function of the automaton to ``jump'' to any letter in the input and consume it~\cite{meduna2012jumping}. For infinite words, this definition does not capture the (arguably) correct notion of jumping automata, since an automaton could skip letters without ever returning to them, e.g., for input $(ab)^\omega$ read only $a^\omega$ by jumping over all the $b$'s. Under our definition, the latter is not allowed, since $a^\omega\not \sim (ab)^\omega$. 
    Clearly, however, one can still obtain our ``Jumping'' definition by considering permutations of words, rather than a jumping reading head.
\end{remark}

\section{Jumping Languages}
\label{sec:JBW}
In this section we study the properties of $\jL(\cA)$ for an automaton $\cA$. We start by characterizing $\jL(\cA)$ using an extension of semilinear sets.

\subsection{Jumping Languages and Masked Semilinear Sets}
\label{subsec:JBA and semilinear}
Recall that $\bbNinf=\bbN_\infty^\Sigma\setminus \bbN^\Sigma$ and consider a vector $\vec{v}\in \bbNinf$. We separate the $\infty$ coordinates of $\vec{v}$ from the finite ones by, intuitively, writing $\vec{v}$ as a sum of a vector in $\bbN^\Sigma$ and a vector from $\maskSetExplicit$.
Formally, let $\maskSet=\maskSetExplicit$ be the set of \emph{masks}, namely vectors with entries $0$ and $\infty$ that are not all $0$, we refer to each $\mask\in \maskSet$ as a \emph{mask}. We denote by $\mask|_0=\{\sigma\in \Sigma\mid \mask(\sigma)=0\}$ and $\mask|_\infty=\{\sigma\in \Sigma\mid \mask(\sigma)=\infty\}$ the 0 and $\infty$ coordinates of $\mask$, respectively.

For $\vec{x}\in \bbN^\Sigma$ and $\mask\in \maskSet$ let $\vec{x}\oplus \mask\in \bbNinf$ be the vector such that for all $\sigma\in \Sigma$ $(\vec{x}+\mask)_\sigma=\vec{x}_\sigma$ if $\sigma\in \mask|_0$ and $(\vec{x}+\mask)_\sigma=\infty$ if $\sigma\in \mask|_\infty$. 
Note that every $\vec{v}\in \bbNinf$ can be written as $\vec{x}+\mask$ where $\vec{x}$ is obtained from $\vec{v}$ be replacing $\infty$ with 0 (or indeed with any number), and having $\mask$ match the $\infty$ coordinates of $\vec{v}$.

We now augment semilinear sets with masks. A \emph{masked semilinear set} is a union of the form $S=\bigcup_{\mask\in \maskSet} S_\mask+\mask$ where $S_\mask$ is a semilinear set for every $\mask$. 
We also interpret $S$ as a subset of $\bbNinf$ by interpreting addition as adding $\mask$ to each vector in the set $S_\mask$. Note that the union above is disjoint, since adding distinct masks always results in distinct vectors.

Consider a mask $\mask\in \maskSet$. Two vectors $\vec{u},\vec{v}\in \bbN^\Sigma$ are called \emph{$\mask$-equivalent} if $\vec{u}+\mask=\vec{v}+\mask$ (i.e., if they agree on $\mask|_0$). 
We say that a semilinear set $R\subseteq \bbN^\Sigma$ is \emph{$\mask$-oblivious} if for every $\mask$-equivalent vectors $\vec{u},\vec{v}$ we have $\vec{u}\in R\iff \vec{v}\in R$. Intuitively, an $\mask$-oblivious set does not ``look'' at $\mask|_\infty$.

We say that the masked semilinear set $S$ above is \emph{oblivious} if every $S_\mask$ is $\mask$-oblivious. Note that the property of being oblivious refers to a specific representation of $S$ with semilinear sets $S_\mask$ for every mask $\mask$.
In a way, it is natural for a masked semilinear set to be oblivious, since semantically, adding $\mask$ to a vector in $S_\mask$ already ignores the $\infty$ coordinates of $\mask$, so $S_\mask$ should not ``care'' about them. 
Moreover, if a set $S_\mask$ is $\mask$-oblivious, then in each of its linear components we can, intuitively, partition its period vectors to two types: those that have $0$ in all the masked coordinates, and the remaining vectors that allow attaining any value in the masked coordinates. We refer to this as a \emph{canonical $\mask$-oblivious representation}, as follows.
\begin{definition}
\label{def:canonical}
    A semilinear set $R\subseteq \bbN^\Sigma$ is in \emph{canonical $\mask$-oblivious form} for a mask $\mask$ if $R=\bigcup_{i=1}^k \lin(\vec{b_i},P_i)$ such that for every $1\le i\le k$ we have $P_i=Q_i\cup E_{\mask}$ with the following conditions:
    \begin{itemize}
        \item $\vec{b_i}(\sigma)=0$ for every $\sigma\in \mask|_\infty$.
        \item $\vec{p}(\sigma)=0$ for every $p\in Q_i$ and $\sigma\in \mask|_\infty$.
        \item $E_{\mask}=\{\vec{e_\sigma}\mid \sigma\in \mask|_\infty\}$ with $\vec{e_\sigma}(\tau)=1$ if $\sigma=\tau$ and $0$ otherwise.
    \end{itemize}
\end{definition}

We now show that every masked semilinear set can be translated to an equivalent one in canonical oblivious form (see example in \cref{fig:canonization}). 
\begin{figure}[ht]
    \centering
    \scriptsize
    \begin{tikzpicture}[auto,node distance=2.8cm,scale=1]
    \node (SL1) [draw=none] at (0,0) {
    $\lin\left(\begin{pmatrix}
    1\\0\\\vec{9}\\\vec{4}
    \end{pmatrix},\left\{\begin{pmatrix}
    1\\2\\\vec{5}\\\vec{7}
    \end{pmatrix},\begin{pmatrix}
    1\\0\\\vec{1}\\\vec{3}
    \end{pmatrix}\right\}\right)+\begin{pmatrix}
    0\\0\\\infty\\\infty
    \end{pmatrix}$
    };
    \node (arr) [draw=none, right = 2pt of SL1] {$\rightsquigarrow$};
    \node (SL2) [draw=none, right = 2pt of arr] {
    $\lin\left(\begin{pmatrix}
    1\\0\\0\\0
    \end{pmatrix},\left\{\begin{pmatrix}
    1\\2\\0\\0
    \end{pmatrix},\begin{pmatrix}
    1\\0\\0\\0
    \end{pmatrix},\begin{pmatrix}
    0\\0\\1\\0
    \end{pmatrix},\begin{pmatrix}
    0\\0\\0\\1
    \end{pmatrix}\right\}\right)+\begin{pmatrix}
    0\\0\\\infty\\\infty
    \end{pmatrix}$
    };
\end{tikzpicture}
    \caption{On the left, a masked semilinear set (actually linear), and on the right an equivalent oblivious canonical form. The bold numbers on the left get masked away, so are replaced by $0$.}
    \label{fig:canonization}
\end{figure}

\begin{lemma}
    \label{lem:oblivious}
    Consider a masked semilinear set $S=\bigcup_{\mask\in \maskSet} S_\mask+\mask$. Then there exists an oblivious masked semilinear set $T=\bigcup_{\mask\in \maskSet} T_\mask+\mask$ that represents the same subset of $\bbNinf$, and every $T_\mask$ is in canonical $\mask$-oblivious form. 
\end{lemma}
\begin{proof}
    Intuitively, we obtain the oblivious representation by adding to each $S_\mask$ vectors where the $\infty$-coordinates of $\mask$ can take any value.

    Formally, for every mask $\mask\in \maskSet$ define 
    \[T_\mask=\{\vec{v}\in \bbN^\Sigma \mid \exists \vec{u}\in S_\mask\text{ such that }\vec{u}+\mask=\vec{v}+\mask\}\]
    That is, $T_\mask$ is obtained from $S_\mask$ by including every vector that is $\mask$-equivalent to some vector in $S_\mask$. 

    By construction, we have that $T_\mask$ is $\mask$-oblivious. Moreover, it is easy to see that $S_\mask+\mask=T_\mask+\mask$. Indeed, $S_\mask\subseteq T_\mask$, thus giving one containment, and for every $\vec{x}\in T_{\mask}+\mask$, write $\vec{x}=\vec{v}+\mask$ for some $\vec{v}\in T_\mask$, then there exists some $\vec{u}\in S_\mask$ with $\vec{v}+\mask=\vec{u}+\mask\in S_\mask+\mask$.
    
    It remains to show that $T_\mask$ is semilinear and represent it in canonical form. 
    To this end, for every vector $\vec{p}\in \bbN^\Sigma$, define $\vec{p}|_\mask$ by $\vec{p}|_\mask(\sigma)=\vec{p}(\sigma)$ if $\sigma\in \mask|_0$ and $\vec{p}|_\mask(\sigma)=0$ if $\sigma\in \mask|_\infty$. That is, we replace all the $\infty$ coordinates of $\mask$ in $\vec{p}$ by $0$. We lift this to sets of vectors: for a set $P$, let $P|_\mask=\{\vec{p}|_\mask\mid \vec{p}\in P\}$. 

    Now, write $S_\mask=\bigcup_{i\in I}\lin(\vec{b_i},P_i)$, we claim the following:
    \[T_\mask=\bigcup_{i\in I}\lin(\vec{b_i}|_\mask,P_i|_\mask \cup E_\mask)\]
    That is, $T_\mask$ is obtained by zeroing the $\infty$ coordinates of $\mask$ in $S_\mask$, and then arbitrarily adding numbers in these coordinates, using $\vec{e_\sigma}\in E_\mask$ (as per \cref{def:canonical}).
    
    Let $\vec{v}\in T_\mask$, then there exists $\vec{u}\in S_\mask$ such that $\vec{v}+\mask=\vec{u}+\mask$. That is, $\vec{u}(\sigma)=\vec{v}(\sigma)$ for all $\sigma\in \mask|_0$. Since $\vec{u}\in S_\mask$, we can write $\vec{u}=\vec{b_i}+\lambda_1p_1+\ldots+\lambda_k p_k$ for some $i\in I$, $p_1,\ldots,p_k\in P_i$ and $\lambda_1,\ldots,\lambda_k\in \bbN$.
    It then follows that 
    \[\vec{v}=\vec{b_i}|_\mask+\lambda_1p_1|_\mask+\ldots+\lambda_k p_k|_\mask+\sum_{\sigma\in \mask|_\infty} \vec{v}(\sigma)\vec{e_\sigma}\]
    and the latter form is in $\lin(\vec{b_i}|_\mask,P_i|_\mask \cup E_\mask)$.

    Conversely, let $\vec{v}\in \lin(\vec{b_i}|_\mask,P_i|_\mask \cup E_\mask)$, and write 
    \[\vec{v}=\vec{b_i}|_\mask+\lambda_1p_1|_\mask+\ldots+\lambda_k p_k|_\mask+\sum_{\sigma\in \mask|_\infty} \beta_\sigma \vec{e_\sigma}\]
    Define $\vec{u}= \vec{b_i}+\lambda_1p_1+\ldots+\lambda_k p_k$, then $\vec{u}\in S_\mask$ and $\vec{u}+\mask=\vec{v}+\mask$ (since the only coordinates where $\vec{v}$ differs from $\vec{u}$ are those in $\mask|_\infty$), so $\vec{v}\in T_\mask$, and we are done.

    We conclude by defining $T=\bigcup_{\mask\in \maskSet}T_\mask +\mask$, which is an oblivious masked semilinear set in canonical form, as required.
    \qed
\end{proof}

\begin{complexitya}[of \cref{lem:oblivious}]
    \label{comp:SL_to_oblivious}
    The description of $T$ involves introducing at most $|\Sigma|$ vectors to the periods of each semilinear set, and changing some entries to $0$ in the existing vectors. Thus, the description of $T$ is of polynomial size in that of $S$.
\end{complexitya}

Our main result in this section is that the Parikh images of jumping languages coincide with masked semilinear sets. We prove this in the following lemmas.
\begin{lemma}
\label{lem:JBA to semilinear}
Let $\cA$ be an automaton, then $\parikh(\jL(\cA))$ is a masked semilinear set. Moreover, we can effectively compute a representation of it from that of $\cA$.
\end{lemma}

\begin{proof}
Consider $\cA=\tup{\Sigma,Q,\delta,Q_0,\alpha}$. By \cref{def:jumping} we have that $\parikh(\jL(\cA))=\parikh(\regL(\cA))$. Indeed, both sets contain exactly the Parikh images of words in $\regL(\cA)$. Recall (or see \cref{apx:omega_regular_union}) that every $\omega$-regular language can be (effectively) written as a union of the form $\regL(\cA)=\bigcup_{i=1}^{m} S_i\cdot T_i^\omega$ where $S_i,T_i\subseteq \Sigma^*$ are regular languages over finite words.

Intuitively, we will show that $\parikh(T_i^\omega)$ can be separated to letters that are seen finitely often, and those that are seen infinitely often, where the latter will induce the mask.
To this end, for every $\emptyset\neq \Gamma\subseteq \Sigma^*$ define $\mask_\Gamma\in \maskSet$ by setting $\mask_\Gamma(\sigma)=\infty$ if $\sigma\in \Gamma$ and $\mask_\Gamma(\sigma)=0$ if $\sigma\notin \Gamma$. Now, for every regular language $T_i$ in the union above define 
\[I(T_i)=\{\mask_\Gamma \mid \forall \sigma\in \Gamma\  \exists w\in T_i\cap \Gamma^* \text{ s.t. }\parikh(w)(\sigma)>0\}\]
That is, $I(T_i)$ is the set of masks $\mask_\Gamma$ such that every letter in $\Gamma$ occurs in some word in $T_i$ that contains only letters from $\Gamma$. Intuitively, $\mask_\Gamma\in I(T_i)$ can be attained by an infinite concatenation of words from $T_i$ by covering all letters in $\Gamma$ whilst not using letters outside of $\Gamma$.

We now claim that for every $1\le i\le m$ it holds that $\parikh(S_i\cdot T_i^\omega)=\parikh(S_i\cdot T_i^*)+I(T_i)$. To reduce clutter, we drop the subscript $i$ for this proof.
Consider $\vec{v}\in \parikh(S\cdot T^\omega)$, then there exists a word $w=x_1x_2\cdots $ such that $\parikh(w)=\vec{v}$ with $x_1\in S$ and $x_j\in T$ for all $j>1$. Let $\Gamma$ be the set of letters that occur infinitely often in $w$, and let $N_0\ge 2$ be such that for all $n\ge N_0$ we have $x_n\in \Gamma^*$. Observe that $\mask_\Gamma\in I(T)$ since for every $\sigma\in \Gamma$ there exists some word $x_j\in T_i\cap \Gamma^*$ (for $j\ge N_0$) such that $\sigma$ occurs in $x_j$. Thus, $\parikh(x_{N_0}x_{N_0+1}\cdots)=\mask_\Gamma$. Moreover, $x_1\cdots x_{N_0-1}\in S\cdot T^*$, so $\vec{v}=\parikh(w)=\parikh(x_1\cdots x_{N_0-1})+\mask_\Gamma\in \parikh(S\cdot T^*)+I(T)$.

Conversely, let $\vec{v}\in \parikh(S\cdot T^*)+I(T)$, then there exist $y\in S\cdot T^*$ and $\mask_\Gamma\in I(T)$ for some $\Gamma\subseteq \Sigma$ such that $\vec{v}=\parikh(y)+\mask_\Gamma$. Since $\mask_\Gamma\in I(T)$, there exists a set of words $\{z_1,\ldots ,z_r\}\subseteq T\cap \Gamma^*$ such that every $\sigma\in \Gamma$ occurs in some word in this set. We thus have $\parikh((z_1\cdots z_r)^\omega)=\mask_\Gamma$ and $(z_1\cdots z_r)^\omega\in T^\omega$, so $y\cdot (z_1\cdots z_r)^\omega\in S\cdot T^\omega$ and we have
$\vec{v}=\parikh(y)+\mask_\Gamma=\parikh(y\cdot (z_1\cdots z_r)^\omega)\in \parikh(S\cdot T^\omega)$, concluding our claim.

It follows that $\parikh(\jL(\cA))=\bigcup_{i=1}^k\parikh(S_i\cdot T_i^*)+I(T_i)$. We can now rearrange the sum by masks instead of by index $i$: for every mask $\mask\in \maskSet$ write $S_\mask=\bigcup_{i: \mask\in I(T_i)} \parikh(S_i\cdot T_i^*)$ (note that $S_\mask$ could be empty). Since $S_i\cdot T_i^*$ is a regular language, then by Parikh's theorem~\cite{parikh1966context} its Parikh image is semilinear, so $S_\mask$ is semilinear. Thus, $\parikh(\jL(\cA))=\bigcup_{\mask\in \maskSet} S_\mask +\mask$ is a masked semilinear set.
\qed
\end{proof}
\begin{complexitya}[of \cref{lem:JBA to semilinear}]
Writing $\regL(\cA)$ as a union involves only a polynomial blowup (see~\cref{apx:omega_regular_union}). Then, we split the resulting expression to a union over $2^{|\Sigma|}$ masks. Within the union, we convert a nondeterministic automaton for $S_i\cdot T_i^*$ to a semilinear set. By \cite[Theorem 4.1]{to2010parikh} (also~\cite{kopczynski2010parikh}), the resulting semilinear set has description size polynomial in the number of states $n$ of $\cA$ and singly-exponential in $|\Sigma|$. Moreover, the translation can be computed in time $2^{O(|\Sigma|^2\log (n|\Sigma|))}$.
\end{complexitya}

For the converse direction, we present a stronger result, namely that we can construct a \emph{deterministic} automaton to capture a masked semilinear set.

\begin{lemma}
\label{lem:semilinear to JBA}
Consider a masked semilinear set $S$, then there exists a deterministic automaton $\cA$ such that $\parikh(\jL(\cA))=S$.
\end{lemma}
\begin{proof}
By \cref{lem:oblivious}, we can assume $S=\bigcup_{\mask\in \maskSet} S_\mask+\mask$ and that every $S_\mask$ is in canonical $\mask$-oblivious form.
Let $\mask\in \maskSet$ and write $S_\mask=\bigcup_{i=1}^k \lin(\vec{b_i},P_i\cup E_\mask)$ as per \cref{def:canonical}. Consider a linear set $L=\lin(\vec{b},P\cup E_\mask)$ in the union above, and we omit the subscript $i$ for brevity. We start by constructing a deterministic automaton $\cD$ such that $\parikh(\jfinL(\cD))=L$. To this end, let $w_{\vec{b}}\in \Sigma^*$ be a word such that $\parikh(w_{\vec{b}})=\vec{b}$ and similarly for every $\vec{p}\in P$ let $w_{\vec{p}}\in \Sigma^*$ such that $\parikh(w_{\vec{p}})=\vec{p}$. Note that $\parikh(\sigma)=\vec{e_\sigma}$ for every $\sigma\in \mask|_\infty$.

Let $\cD'$ be a deterministic automaton for the regular expression $w_{\vec{b}}\cdot (w_{\vec{p_1}}+\ldots+w_{\vec{p_n}})^*$ where $P=\{\vec{p_1},\ldots,\vec{p_n}\}$. Next, let $w_{E}=\sigma_1\cdots \sigma_k$ be a word obtained by concatenating all the letters in $\mask|_\infty$ in some order. We obtain $\cD$ from $\cD'$ by connecting every accepting state of $\cD'$, upon reading $\sigma_1$, to a cycle that allows reading $w_{E}^\omega$. The  accepting states of $\cD$ are those on the $w_{E}$ cycle. 
Crucially, the transition from $\cD'$ upon reading $\sigma_1$ retains determinism, since $P$ is in canonical form, and therefore $\vec{p}(\sigma_1)=0$ for every $\vec{p}\in P$. That is, the letter $\sigma_1$ is not seen in any transition of $\cD'$, allowing us to use it in the construction of $\cD$.
The construction is demonstrated in \cref{fig:det_for_semilinear}. 
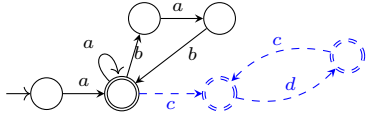
\begin{figure}
    \centering
    \scriptsize
    \begin{tikzpicture}[auto,node distance=2.8cm,scale=1]
        \tikzset{every state/.style={minimum size=12pt, inner sep=1}};
        \node (q0) [initial, state, initial text = {}] at (0,0) {};
        \node (q1) [state,accepting] at (1,0) {};
        \node (q2) [state] at (1.3,1) {};
        \node (q3) [state] at (2.3,1) {};
        \node (q4) [state,blue,accepting,dashed] at (2.3,0) {};
        \node (q5) [state,blue,accepting,dashed] at (4,0.5) {};
        \path [-stealth]
        (q0) edge node [above] {$a$} (q1)
        (q1) edge [in=100,out=135,loop] node {$a$} (q1)
        (q1) edge node [below,pos=0.9] {$b$} (q2)
        (q2) edge node [above,pos=0.4] {$a$} (q3)
        (q3) edge node [below,pos=0.2] {$b$} (q1)
        (q1) edge[blue,dashed] node [below] {$c$} (q4)
        (q4) edge[bend right,blue,dashed] node [above] {$d$} (q5)
        (q5) edge[bend right,blue,dashed] node [above] {$c$} (q4);
    \end{tikzpicture}
    \caption{The automaton $\cD$ for the linear set from \cref{fig:canonization} (over $\Sigma=\{a,b,c,d\}$), with the representative words $w_{\vec{b}}=a$ and $\{bab,a,c,d\}$ for the period. The blue (dashed) parts are the addition of $w_E$ and change of accepting states to obtain $\cD$ from $\cD'$.}
    \label{fig:det_for_semilinear}
\end{figure}

By construction, we have that $\regL(\cD)=w_{\vec{b}}\cdot (w_{\vec{p_1}}+\ldots+w_{\vec{p_n}})^*\cdot w_E^\omega$.
Since $\parikh(w_E^\omega)=\mask$, we now have \[ \parikh(\jL(\cD))=\parikh(\regL(\cD))=\parikh(w_{\vec{b}}\cdot(w_{\vec{p_1}}+\ldots+w_{\vec{p_n}})^*\cdot w_E^\omega)=\lin(\vec{b},P)+\mask=L+\mask\]
where the last equality is because adding $\mask$ to a vector in $L$ masks any addition of vectors in $E_\mask$ by $\infty$.

Next, we observe that for two deterministic automata $\cD_1,\cD_2$, we can construct a deterministic automaton $\cD_3$ such that $\jL(\cD_3)=\jL(\cD_1)\cup \jL(\cD_2)$. Indeed, the standard product construction (where the accepting states are pairs of states where either $\cD_1$ or $\cD_2$ accept) preserves determinism. Thus, we can take products to first capture $S_\mask=\bigcup_{i=1}^k \lin(\vec{b_i},P_i\cup E_\mask)$ and then $S=\bigcup_{\mask\in \maskSet} S_\mask+\mask$, as desired.
\qed
\end{proof}

\begin{complexitya}[of \cref{lem:semilinear to JBA}]
    \label{comp:SL_to_JBA}
    We can naively construct an automaton for the regular expression $w_{\vec{b}}\cdot (w_{\vec{p_1}}+\ldots+w_{\vec{p_n}})^*$ of size $\norm{\lin(\vec{b},P)}=\norm{\vec{b}}+\sum_{i=1}^n\norm{\vec{p_i}}$, where e.g., $\norm{\vec{b}}$ is the sum of entries in $\vec{b}$ (i.e., unary representation). Then, determinization can yield a \DFA of size singly exponential. Then, we take unions by the product construction, where the number of copies corresponds to the number of linear sets in $S$, and then to $2^{|\Sigma|}-1$ many masks. This results in a \DFA of size singly exponential in the size of $S$ and doubly-exponential in $|\Sigma|$ (assuming $S$ is represented in unary). Finally, adding the cycle for the mask adds at most $2^{O(|\Sigma|)}$ states (both to a deterministic or nondeterministic automaton).
\end{complexitya}

Combining \cref{lem:JBA to semilinear,lem:semilinear to JBA} we have the following.
\begin{theorem}
\label{thm:JBA equiv semilinear}
A set $S$ is a masked semilinear set if and only if there exists an automaton $\cA$ such that $S=\parikh(\jL(\cA))$.
\end{theorem}

\subsection{Jumping Languages -- Closure Properties}
\label{subsec:closure_JBA}
Using the characterizations obtained in \cref{subsec:JBA and semilinear}, we can now obtain several closure properties of jumping languages.
First, we remark that jumping languages are clearly closed under union, by taking the union of the automata and nondeterministically choosing which one to start at. We proceed to tackle intersection and complementation. 

Note that the standard intersection construction of B\"uchi automata does not capture the intersection of jumping languages, as demonstrated in \cref{fig:intersection_NBW_fails}.
\begin{figure}[ht]
    \centering
    \begin{tikzpicture}[auto,node distance=2.8cm,scale=1]
    \node (A) [draw=none] at (-1,0.3) {$\cA$:};
    \node (q0) [initial, state, initial text = {},inner sep=0pt, minimum size=15pt] {$q_{0}$};
    \node (q1) [accepting, state,inner sep=0pt, minimum size=15pt] at (2,0) {$q_{1}$};
    \path [-stealth]
    (q0) edge [bend left=20] node {$a$} (q1)
    (q1) edge [bend left=20] node {$b$} (q0);
\end{tikzpicture}\qquad\qquad
    \begin{tikzpicture}[auto,node distance=2.8cm]
    \node (A) [draw=none] at (-1,0.3) {$\cB$:};
    \node (q0) [initial, state, initial text = {},inner sep=0pt, minimum size=15pt] {$s_{0}$};
    \node (q1) [accepting, state,inner sep=0pt, minimum size=15pt] at (2,0) {$s_{1}$};
    \path [-stealth]
    (q1) edge [bend left=20] node {$a$} (q0)
    (q0) edge [bend left=20] node {$b$} (q1);
\end{tikzpicture}
    \caption{The automata $\cA$ and $\cB$ satisfy $\regL(\cA)=\{(ab)^\omega\}$ and $\regL(\cB)=\{(ba)^\omega\}$. Thus, $\jL(\cA)=\jL(\cB)=\{w\in \Sigma^\omega\mid w\text{ has infinitely many $a$'s and $b$'s}\}$. However, if we start by taking the standard intersection of $\cA$ and $\cB$, we would end up with the empty language.}
    \label{fig:intersection_NBW_fails}
\end{figure}
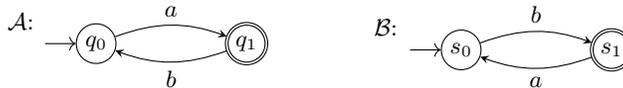
\begin{proposition}
\label{prop:JBA intersection}
Let $\cA$ and $\cB$ be automata, then we can effectively construct an automaton $\cC$ such that $\jL(\cC)=\jL(\cA)\cap \jL(\cB)$.
\end{proposition}
\begin{proof}
Consider a word $w\in \Sigma^\omega$, and observe that $w\in\jL(\cA)\cap \jL(\cB)$ if and only if $\parikh(w)\in\parikh(\jL(\cA))\cap\parikh(\jL(\cB))$. Thus, it suffices to prove that there exists automaton $\cC$ such that $\parikh(\jL(\cC))=\parikh(\jL(\cA))\cap \parikh(\jL(\cB))$.
By \cref{thm:JBA equiv semilinear} we can write $\parikh(\jL(\cA))=\bigcup_{\mask\in\maskSet}S_\mask +\mask$ and $\parikh(\jL(\cB))=\bigcup_{\mask\in\maskSet}T_\mask +\mask$. Furthermore, by \cref{lem:oblivious} we can assume that these are oblivious masked semilinear sets.

We claim that $\parikh(\jL(\cA))\cap\parikh(\jL(\cB))=\bigcup_{\mask\in\maskSet}S_\mask\cap T_\mask +\mask$. That is, we can take intersection by intersecting each pair of semilinear sets, while keeping the masks. Intuitively, this holds because every word uniquely determines the mask it can use to join the union, so intersecting the unions amounts (using obliviousness) to intersecting the set of finite-valued vectors corresponding to the finite parts of the Parikh images of words in the language.

Indeed, let $\vec{x}\in\bigcup_{\mask\in\maskSet}S_\mask\cap T_\mask +\mask$, then there exists $\mask\in\maskSet$ and $\vec{v}\in S_\mask\cap T_\mask$ such that $\vec{x}=\vec{v}+\mask$. In particular, $\vec{v}+\mask\in \parikh(\jL(\cA))$ and $\vec{v}+\mask\in \parikh(\jL(\cB))$, so $\vec{x}=\vec{v}+\mask\in \parikh(\jL(\cA))\cap\parikh(\jL(\cB))$.

Conversely, let $\vec{x}\in\parikh(\jL(\cA))\cap\parikh(\jL(\cB))$. Recall that $\vec{x}$ uniquely defines a mask $\mask$, so there exist $\vec{u}\in S_\mask$, $\vec{v}\in T_\mask$ such that $\vec{x}=\vec{u}+\mask=\vec{v}+\mask$. It follows that $\vec{u}$ and $\vec{v}$ are $\mask$-equivalent. Since $S_\mask$ and $T_\mask$ are $\mask$-oblivious, we also have that $\vec{v}\in S_\mask$ and $\vec{u}\in T_\mask$. In particular, $\vec{v}\in S_\mask\cap T_\mask$ so $\vec{x}=\vec{v}+\mask\in\bigcup_{\mask\in\maskSet}S_\mask\cap T_\mask +\mask$.

Finally, semilinear sets are closed under intersection~\cite{beierdescriptional,chistikov2016taming}. Thus, $S_\mask\cap T_\mask$ is semilinear for every $\mask$, so $\bigcup_{\mask\in\maskSet}S_\mask\cap T_\mask +\mask$ is a semilinear masked set and from \cref{lem:semilinear to JBA} there exists an automaton $\cC$ such that $\jL(\cC)=\jL(\cA)\cap (\jL(\cB)$.
\qed
\end{proof}
\begin{complexitya}[of \cref{prop:JBA intersection}]
    \label{comp:JBA_intersection}
    The description size of the intersection of semilinear sets is polynomial in the description size of the entries of the vectors and the size of the union, and singly exponential in $|\Sigma|$~\cite{beierdescriptional,chistikov2016taming}. Since the blowup in \cref{lem:JBA to semilinear,lem:oblivious} is polynomial, then the only significant blowup is \cref{lem:semilinear to JBA}, due to the construction of an automaton whose size is exponential in $|\Sigma|$. However, it is not hard to see that the two exponential blowups in $|\Sigma|$ are orthogonal, and can be merged to a singly-exponential blowup.
\end{complexitya}

\begin{proposition}
\label{prop:JBA complement}
Consider an automaton $\cA$. There exists an automaton $\cB$ such that $\jL(\cB)=\Sigma^\omega\setminus \jL(\cA)$.
\end{proposition}

\begin{proof}
Observe that $w\in \Sigma^\omega\setminus \jL(\cA)$ if and only if $\parikh(w)\in \bbNinf\setminus \parikh(\jL(\cA))$. Indeed, $w\in \Sigma^\omega\setminus \jL(\cA)$ means that every permutation of $w$ is not accepted by $\cA$, i.e., the Parikh image of $w$ is not a Parikh image of any word accepted by $\cA$. Thus, it suffices to construct $\cB$ such that $\parikh(\jL(\cB))=\bbNinf\setminus \parikh(\jL(\cA))$.

By \cref{thm:JBA equiv semilinear}, we can write $\parikh(\jL(\cA))=\bigcup_{\mask\in\maskSet}S_\mask +\mask$. Moreover, by \cref{lem:oblivious} we can assume that this is an oblivious masked semilinear set.
We claim that $\bbNinf\setminus \parikh(\jL(\cA))=\bigcup_{\mask\in\maskSet} (\bbN^\Sigma\setminus S_\mask) +\mask$. That is, we can complement the union by complementing each semilinear set, while keeping the masks. 

To show this, consider a word $w\in \Sigma^\omega$ and write $\parikh(w)=\vec{x_w}+\mask_w$ where $\vec{x_w}\in \bbN^\Sigma$, as per \cref{subsec:JBA and semilinear}.

For the first direction, assume $\vec{x_w}+\mask_w \in \bbNinf\setminus \parikh(\jL(\cA))$, then $\vec{x_w}+\mask_w\notin \bigcup_{\mask\in\maskSet}S_\mask +\mask$. In particular, $\vec{x_w}\notin S_{\mask_w}$ (otherwise we would have $\vec{x_w}+\mask_{w}\in S_{\mask_w}+\mask_w$, which is part of the union). It follows that $\vec{x_w}+\mask_{w}\in (\bbN^\Sigma\setminus S_{\mask_w}) +\mask_w$, so $\parikh(w)\in \bigcup_{\mask\in\maskSet} (\bbN^\Sigma\setminus S_\mask) +\mask$.

Conversely, assume $\parikh(w)=\vec{x_w}+\mask_w\in \bigcup_{\mask\in\maskSet} (\bbN^\Sigma\setminus S_\mask) +\mask$, then observe that $\vec{x_w}+\mask_w\in (\bbN^\Sigma\setminus S_{\mask_w}) +\mask_w$. Indeed, no other part of the union has the mask $\mask_w$. 
Assume by way of contradiction that $\vec{x_w}+\mask_w\in \parikh(\jL(\cA))=\bigcup_{\mask\in\maskSet}S_\mask +\mask$, then by the same reasoning, $\vec{x_w}+\mask_w\in S_{\mask_w}+\mask_w$. Thus, there are two vectors $\vec{y}\notin S_{\mask_w}$ and $\vec{z}\in S_{\mask_w}$ such that $\vec{y}+\mask=\vec{x_w}+\mask=\vec{z}+\mask$. This contradicts the assumption that $S_{\mask_w}$ is oblivious (note that if $S_{\mask_w}$ is not oblivious, this is not a contradiction).

We conclude that $\bbNinf\setminus \parikh(\jL(\cA))=\bigcup_{\mask\in\maskSet} (\bbN^\Sigma\setminus S_\mask) +\mask$, and since semilinear sets are closed under complementation~\cite{beierdescriptional,chistikov2016taming}, the latter is also a masked semilinear set. By applying \cref{thm:JBA equiv semilinear} in the converse direction, we conclude that there exists an automaton $\cB$ such that $\parikh(\jL(\cB))=\bbNinf\setminus \parikh(\jL(\cA))$, and as mentioned above, for such $\cB$ we have $\jL(\cB)=\Sigma^\omega\setminus \jL(\cA)$.
\qed
\end{proof}

\begin{complexitya}[of \cref{prop:JBA complement}]
    \label{comp:complementation_JBA}
    The description size of the complement of a semilinear set is singly exponential, in both the degree and the description of the entries~\cite{beierdescriptional,chistikov2016taming}. We proceed similarly to \cref{comp:JBA_intersection}, ending up with a singly-exponential blowup (not just in $|\Sigma|$).
\end{complexitya}

Recall that B\"uchi automata are generally not determinizable. That is, there exists an automaton $\cA$ such that $\regL(\cA)$ is not recognizable by a deterministic B\"uchi automaton~\cite{landweber1968decision}. In stark contrast, an immediate corollary of \cref{lem:JBA to semilinear,lem:semilinear to JBA} is that jumping languages do admit determinization (also see \cref{comp:SL_to_JBA}). 
\begin{proposition}
    \label{prop:jumping_determinizable}
    For every automaton $\cA$ there exists a deterministic automaton $\cD$ such that $\jL(\cA)=\jL(\cD)$.
\end{proposition}
In particular, this means that additional acceptance conditions (e.g., Rabin, Streett, Muller and parity) cannot add expressiveness to the model.
\begin{remark}
    \label{rmk:jumping_not_really_det}
    Discussing determinization of jumping languages is slightly misleading: the definition of acceptance in a jumping language asks that \emph{there exists} some permutation that is accepted by the automaton. This existential quantifier can be thought of as ``semantic'' nondeterminism. Thus, in a way, jumping languages are inherently nondeterministic, regardless of the underlying syntactic structure.     
\end{remark}

Finally, as is the case for jumping languages over finite words~\cite{meduna2012jumping,fernau2015jumping,fernau2017characterization}, jumping languages and $\omega$-regular languages are incomparable. Indeed, the automata in~\cref{fig:intersection_NBW_fails} are examples of B\"uchi automata whose languages are not permutation-invariant, and are hence not jumping languages, and in \cref{fig:JBA_no_equiv_NBW} we demonstrate that the converse also does not hold.
\begin{figure}[ht]
    \centering
    \begin{tikzpicture}[auto,node distance=2.8cm,scale=1]
    \tikzset{every state/.style={minimum size=15pt, inner sep=3}};
    \node (q0) [initial, state, initial text = {}] {$q_{0}$};
    \node (q1) [state,inner sep=0pt] at (2,0) {$q_{1}$};
    \node (q2) [accepting, state,inner sep=0pt] at (-2,0) {$q_{2}$};
    \path [-stealth]
    (q0) edge [bend left=20] node {$a$} (q1)
    (q1) edge [bend left=20] node [above] {$b$} (q0)
    (q0) edge [bend left=20] node [above] {$c$} (q2)
    (q2) edge [loop left] node {$c$} (q2);
\end{tikzpicture}\qquad\qquad
    \caption{An automaton $\cA$ for which $\jL(\cA)=\{u\cdot c^\omega\mid u\in \{a,b,c\}^* \wedge \numl{u}{a}=\numl{u}{b}\}$, which is not $\omega$-regular by simple pumping arguments.}
    \label{fig:JBA_no_equiv_NBW}
\end{figure}
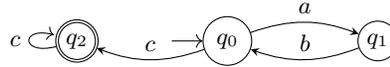

\subsection{Jumping Languages - Algorithmic Properties}
\label{subsec:algorithms_JBA}
We turn our attention to algorithmic properties of jumping languages. We first notice that given an automaton $\cA$ we have that $\jL(\cA)\neq \emptyset$ if and only if $\regL(\cA)\neq \emptyset$. Thus, non-emptiness of $\jL(\cA)$ is \nlc, as it is for B\"uchi automata~\cite{emerson1985modalities,kupferman2018automata}.
Next, we consider the standard decision problems:
\begin{itemize}
    \item \textbf{Containment}: given automata $\cA,\cB$, is $\jL(\cA)\subseteq \jL(\cB)$?
    \item \textbf{Equivalence}: given automata $\cA,\cB$, is $\jL(\cA)= \jL(\cB)$?
    \item \textbf{Universality}: given an automaton $\cA$, is $\jL(\cA)=\Sigma^\omega$?
\end{itemize}
We show that all of these problems reduce to their analogues in finite-word jumping automata. 
\begin{theorem}
    \label{thm:alg_JBA}
    Containment, Equivalence and Universality for jumping languages are $\conp$-complete (for fixed size alphabet).
\end{theorem}
\begin{proof}
We start with containment. by \cref{thm:JBA equiv semilinear,lem:oblivious}, given automata $\cA,\cB$ we obtain canonical oblivious masked semilinear representation:
$\parikh(\jL(\cA))=\bigcup_{\mask\in\maskSet}S_\mask +\mask$ and $\parikh(\jL(\cB))=\bigcup_{\mask\in\maskSet}T_\mask +\mask$.
Then, we have $\jL(\cA)\subseteq \jL(\cB)$ if and only if $\parikh(\jL(\cA))\subseteq \parikh(\jL(\cB))$, which in turn happens if and only if for every mask $\mask$ it holds that $S_\mask+\mask\subseteq T_\mask+\mask$. We claim that the latter holds if and only if $S_\mask\subseteq T_\mask$. 

Indeed, clearly if $S_\mask\subseteq T_\mask$ then  $S_\mask+\mask\subseteq T_\mask+\mask$. Conversely, assume  $S_\mask+\mask\subseteq T_\mask+\mask$ and let $\vec{u}\in S_\mask$, so $\vec{u}+\mask\in T_\mask+\mask$. Then, there exists $\vec{v}\in T_\mask$ such that $\vec{u}+\mask=\vec{v}+\mask$, but since $T_\mask$ is $\mask$-oblivious, this means that $\vec{u}\in T_\mask$.

Thus, deciding whether $\jL(\cA)\subseteq \jL(\cB)$ amounts to deciding whether $S_\mask\subseteq T_\mask$ for every $\mask$. Since the latter is decidable, we get decidability of containment. Moreover, we can obtain in polynomial time nondeterministic automata $\cC^\mask_1,\cC^\mask_2$ corresponding to $S_\mask$ and $T_\mask$ (see \cref{comp:SL_to_oblivious,comp:SL_to_JBA}), so it is enough to decide if $\jfinL(\cC^\mask_1)\subseteq \jfinL(\cC^\mask_2)$ for all $\mask$. The latter problem is \conp-complete~\cite{kopczynski2010parikh,kopczynski2015complexity} (for fixed-size alphabet), and since we can start by nondeterministically guessing $\mask$, we retain this complexity. Moreover, \conp hardness trivially follows (by e.g., reducing a given language $L$ to $L\cdot \#^\omega$ for a new symbol $\#$).

Similar arguments hold for equivalence and universality, concluding the proof.
\qed
\end{proof}

\section{Fixed-Window Jumping Languages}
\label{sec:fixed window}
Recall from \cref{def:jumping} that given an automaton $\cA$ and $k\in \bbN$, the language $\kwinL(\cA)$ consists of the words that have a $k$-window permutation that is accepted by $\cA$. 

Since $k$ is fixed, we can capture $k$-window permutations using finite memory. Hence, as we now show, $\kwinL(\cA)$ is $\omega$-regular.
\begin{theorem}
    \label{thm:WJW to NBW}
    Consider an automaton $\cA$, then for every $k\in \bbN$ there exists a B\"uchi automaton $\cB_k$ such that  $\kwinL(\cA)=\regL(\cB_k)$.
\end{theorem}
\begin{proof}[Sketch]
Intuitively, we construct $\cB_k$  so that it stores in a state the Parikh image of $k$ letters, then nondeterministically simulates $\cA$ on all words with the stored Parikh image, while noting when an accepting state may have been traversed. We illustrate the construction for $k=2$ in \cref{fig:example_WJW_to_NBW}. The details are in \cref{apx:WJW to NBW}.  
\qed 
\end{proof}
\begin{complexitya}[of \cref{thm:WJW to NBW}]
    \label{comp:WJW to NBW}
    If $\cA$ has $n$ states, then $\cB$ has at most $k^{|\Sigma|}\times n+n$ states. Moreover, computing $\cB$ can be done effectively, by keeping for every Parikh image the set of states that are reachable (resp. reachable via an accepting state).
\end{complexitya}
\begin{figure}[ht]
    \centering
    \begin{tikzpicture}[auto,node distance=1.5cm,scale=1]
    \node (D) [draw=none] at (-2,1) {$\cD$:};
    \node (q0) [initial above, state, initial text = {},inner sep=0pt, minimum size=15pt] at (0,0) {$q_{0}$};
    \node (q1) [state,inner sep=0pt, minimum size=15pt] at (-1.5,0) {$q_{1}$};
    \node (q2) [state,inner sep=0pt, minimum size=15pt] at (1.5,0) {$q_{2}$};
    \node (q3) [accepting, state,inner sep=0pt, minimum size=15pt] at (1.5,1.2) {$q_{3}$};
    \path [-stealth]
    (q0) edge [bend left=20] node {$a$} (q1)
    (q1) edge [bend left=20] node {$b$} (q0)
    (q0) edge [bend left=20] node {$b$} (q2)
    (q2) edge [bend left=20] node {$b$} (q0)
    (q2) edge node {$a$} (q3)
    (q3) edge [loop left] node {$b$} (q3);
\end{tikzpicture}\qquad
\begin{tikzpicture}[auto,node distance=2.8cm,scale=1]
    \node (B2) [draw=none] at (-1,0.3) {$\cB_2$:};
    \scriptsize
    \node (q0) [initial, state, initial text = {},inner sep=0pt, minimum size=15pt] at (0,0) {$q_{0}\begin{pmatrix}
        0\\0
    \end{pmatrix}$};
    \node (q0_10) [state,inner sep=0pt, minimum size=15pt] at (2,0.7) {$q_{0}\begin{pmatrix}
        1\\0
    \end{pmatrix}$};
    \node (q0_01) [state,inner sep=0pt, minimum size=15pt] at (2,-0.7) {$q_{0}\begin{pmatrix}
        0\\1
    \end{pmatrix}$};
    \node (q3) [accepting, state,inner sep=0pt, minimum size=25pt] at (4,-0.7) {$q_{3}$};
    \node (q3_01) [state,inner sep=0pt, minimum size=15pt] at (4,0.7) {$q_{3}\begin{pmatrix}
        0\\1
    \end{pmatrix}$};
    \path [-stealth]
    (q0) edge [bend left=20] node {$a$} (q0_10)
    (q0) edge [bend left=5] node[pos=0.7] {$b$} (q0_01)
    (q0_10) edge [bend left=5] node[pos=0.3] {$b$} (q0)
    (q0_10) edge [bend left=5] node[pos=0.2] {$b$} (q3)
    (q0_01) edge [bend left=20] node {$a,b$} (q0)
    (q0_01) edge [bend right=10] node {$a$} (q3)
    (q3) edge [bend left=10] node {$b$} (q3_01)
    (q3_01) edge [bend left=10] node {$b$} (q3);
\end{tikzpicture}
      \caption{An automaton $\cD$, and the construction $\cB_2$ for $\kwinL[2](\cD)$ as per~\cref{thm:WJW to NBW}.}
    \label{fig:example_WJW_to_NBW}
\end{figure}
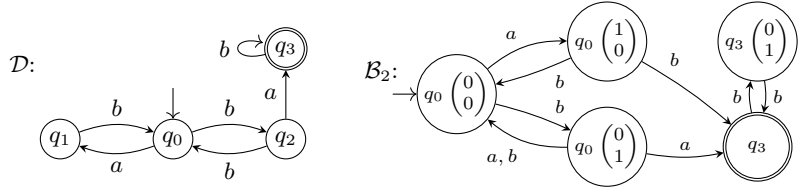

A converse of \cref{thm:WJW to NBW} also holds, in the following sense: we say that a language $L\subseteq \Sigma^\omega$ is $k\boxplus$ invariant if $w\in L\iff w'\in L$ for every $w\simkw w'$. The following result follows by definition.
\begin{proposition}
    \label{prop:NBW_k_inv_WJW}
    Let $\cA$ be an automaton such that $\regL(\cA)$ is $k\boxplus$-invariant, then $\regL(\cA)=\kwinL(\cA)$.
\end{proposition}

\cref{thm:WJW to NBW,prop:NBW_k_inv_WJW} yield that $k$-window jumping languages are closed under the standard operations under which $\omega$-regular languages are closed (union, intersection, complementation), as these properties retain $k\boxplus$ invariance. Similarly, all algorithmic problems can be reduced to the same problems on B\"uchi automata (by \cref{comp:WJW to NBW}, assuming $|\Sigma|$ is fixed).

Notice that the construction in \cref{thm:WJW to NBW} yields a nondeterministic B\"uchi automaton. This is not surprising in light of \cref{rmk:jumping_not_really_det}. It does raise the question of whether we can find a deterministic B\"uchi automaton for $\kwinL(\cD)$ for a \emph{deterministic} automaton $\cD$. In~\cref{prop:kwin_det_no_DBW}, we prove that this is not the case. That is, even deterministic $k$-window jumping has inherent nondeterminism.
\begin{proposition}
    \label{prop:kwin_det_no_DBW}
    There is no nondeterministic B\"uchi automaton whose language is $\kwinL(\cD)$ for $k\ge 2$ and $\cD$ as in \cref{fig:example_WJW_to_NBW}.
\end{proposition}
\begin{proof}
    Consider the deterministic automaton $\cD$ in \cref{fig:example_WJW_to_NBW}. 
    We start by showing that $\regL(\cB_2)=\kwinL[2](\cD)$ cannot be recognized by a deterministic B\"uchi automaton. Below, we show this \emph{mutatis mutandis} for every even $k$, and with slightly more effort -- for all $k\ge 2$.

    By way of contradiction, consider a deterministic automaton $\cA$ such that $\regL(\cA)=\kwinL[2](\cD)$, then $\cA$ accepts $(ba)(bb)^\omega$, and therefore its run passes through an accepting state after some $(ba)(bb)^{m_1}$. Then, $\cA$ also accepts $(ba)(bb)^{m_1}(ab)(bb)^\omega$, so again there is $m_2$ such that $\cA$ reaches an accepting state after $(ba)(bb)^{m_1}(ab)(bb)^{m_2}$, we consider the word $(ba)(bb)^{m_1}(ab)(bb)^{m_2}(ab)(bb)^\omega$ and proceed in this fashion to construct an accepting run on a word with infinitely many $(ab)$, which is not in $\kwinL[2](\cD)$.
   
We now turn to show the same for any $k\ge 2$.

Assume by way of contradiction that there exists $k\in\bbN, k\ge 2$ such that $\kwinL(\cD)=\regL(A)$ for a deterministic automaton $\cA$. $k$ can be either odd or even.

We first assume that $k$ is even and consider the word $w=bab^\omega$. $w\in\kwinL(\cD)$ as there is $w'\simkw w$ that is accepted by $\cD$, namely $w'=w$. Since $bab^\omega\in\regL(\cA)$, the unique run of $\cA$ on it is accepting, and let $m_1$ be an index such that $|bab^{m_1}|$ is divisible by $k$ and $\alpha$ was visited at least once while reading $bab^{m_1}$. Now, consider the word $w=bab^{m_1}bab^\omega$. $w\in\kwinL(\cD)$ as there is $w'\simkw w$ that is accepted by $\cD$, and that is $w'=abb^{m_1}bab^\omega$. $bab^{m_1}bab^\omega\in\regL(\cA)$, then the unique run of $\cA$ on it is accepting, and let $m_2$ be an index such that $|bab^{m_1}bab^{m_2}|$ is divisible by $k$ and $\alpha$ was visited at least twice while reading $bab^{m_1}bab^{m_2}$. Consider the word $w=bab^{m_1}bab^{m_2}bab^\omega$, it is also in $\kwinL(\cD)$ because $w'=abb^{m_1}abb^{m_2}bab^\omega$ is accepted by $\cD$. Following this fashion, we can construct an infinite word $w=bab^{m_1}bab^{m_2}bab^{m_3}...$ such that  $\cA$'s run on it visits $\alpha$ infinitely many often despite the fact that $w\notin\kwinL(\cD)$. 

We then assume that $k\ge 2$ is odd and consider the word $w=ab^\omega$. We have that $w\in\kwinL(\cD)$ as there is $w'\simkw w$ that is accepted by $\cD$, namely $w'=bab^\omega$. Since $ab^\omega\in\regL(A)$, the unique run of $\cA$ on it is accepting, and let $m_1$ be an index $|ab^{m_1}|$ is divisible by $k$ and $\alpha$ was visited at least once while reading $ab^{m_1}$. Now, consider the word $w=ab^{m_1}ab^\omega$. Again $w\in\kwinL(\cD)$ as there is $w'\simkw w$ that is accepted by $\cD$, namely $w'=w$ (with the accepting run $\rho=q_0q_1(q_0q_2)^{k-2}q_3^\omega$). Since  $ab^{m_1}ab^\omega\in\regL(A)$, the unique run of $\cA$ on it is accepting, and let $m_2$ be an index such that $|ab^{m_1}ab^{m_2}|$ is divisible by $k$ and $\alpha$ was visited at least twice while reading $ab^{m_1}ab^{m_2}$. Similarly we can show that $w=ab^{m_1}ab^{m_2}ab^\omega\in\kwinL(\cD)$, and we can proceed and construct an infinite word $w=ab^{m_1}ab^{m_2}ab^{m_3}...$ such that $\cA$'s run on it visits $\alpha$ infinitely many often despite the fact that $w\notin\kwinL(\cD)$. 
It follows that there isn't a deterministic automaton $\cA$ such that $\kwinL(\cD)=\regL(A)$ for any $k>1$.
\qed
\end{proof}

\section{Existential-Window Jumping Languages}
\label{sec:exists_jumping}
Recall from \cref{def:jumping} that given an automaton $\cA$, the language $\ewinL(\cA)$ contains the words that have a $k$-window permutation that is accepted in $\cA$ for some $k$, i.e., $\ewinL(\cA)=\bigcup_{k\in \bbN}\kwinL(\cA)$. We briefly establish some expressiveness properties of $\ewinL(\cA)$. 

Perhaps surprisingly, we show that $\exists\boxplus$ languages are strictly more expressive than Jumping languages.
We start by showing that every jumping language can be defined as the $\exists\boxplus$ language of some automaton.
\begin{theorem}
\label{thm:JBW_is_EWJ}
    Let $\cA$ be an automaton, then there exists an automaton $\cB$ such that $\ewinL(\cB)=\jL(\cA)$. 
\end{theorem}
\begin{proof}
     Intuitively, we modify the construction of the jumping automaton in~\cref{lem:semilinear to JBA} so that it produces an automaton for which the $\exists\boxplus$ language coincides with the given semilinear set.
    
    Consider an automaton $\cA$. By \cref{lem:JBA to semilinear,lem:oblivious}, we can write $\parikh(\jL(\cA))=\bigcup_{\mask\in\maskSet}S_\mask +\mask$ in canonical oblivious form. Similarly to \cref{lem:semilinear to JBA}, consider a specific linear set $L=\lin(\vec{b},P\cup E_\mask)$ in the union above, and consider the words $w_{\vec{b}}\in \Sigma^*$ and  $w_{\vec{p}}\in \Sigma^*$ for $\vec{p}\in P$ as in \cref{lem:semilinear to JBA}. 
    Let $\cN'$ be a nondeterministic automaton such that $\finL(\cN')=w_{\vec{b}}\cdot (w_{\vec{p_1}}+\ldots+w_{\vec{p_n}})^*$. We now obtain an automaton $\cN$ by connecting every accepting state of $\cN'$ to a new nondeterministic automaton $\cN_\mask$ which, as a B\"uchi automaton, accepts the language $L_\mask=\{w\in \mask|_\infty^\omega \mid \text{ every }\sigma\in \mask|_\infty^\omega\text{occurs infinitely often in }w\}$. The accepting states of $\cN$ are those of $\cN_\mask$.
    
    Observe that $L_\mask$ is permutation invariant, and in particular $\ewinL(\cN_\mask)=L_\mask$. Indeed, trivially we have $L_\mask=\regL(\cN_\mask)\subseteq \ewinL(\cN_\mask)$, and for the other direction -- if $w\in \ewinL(\cN_\mask)$, then $w\simkw w'$ for some $w'\in \regL(\cN_\mask)=L_\mask$ and some $k\in \bbN$, but since $L_\mask$ is permutation invariant, and since we have in particular that $w\sim w'$, then $w\in L_\mask$.
    
    Furthermore, $\parikh(\regL(\cN))=\lin(\vec{b},P\cup E_\mask)+\mask$.
    We construct $\cB$ by taking the union over $\mask$ in the masked semilinear set.

    We claim that $\ewinL(\cB)=\jL(\cA)$. 
    For the ``easy'' direction, let $w\in \ewinL(\cB)$, then in particular $\parikh(w)\in \parikh(\regL(\cN))=\lin(\vec{b},P\cup E_\mask)\subseteq S=\parikh(\jL(\cA))$ (for some set in the union), so $w\in \jL(\cA)$ (as $\jL(\cA)$ is permutation invariant).

    For the ``hard'' direction, assume $w\in \jL(\cA)$, then $\parikh(w)\in \lin(\vec{b},P\cup E_\mask)+\mask$ for some set in the union. We construct a word $w'\simkw w$ for some $k\in \bbN$ such that $w'\in \regL(\cB)$, thus showing $w\in \ewinL(\cB)$. Let $k\in \bbN$ be large enough so that $w=u\cdot z$ with $z\in \mask|_\infty^\omega$ and $|u|\le k$. There exists such $k$ since $w$ contains only finitely many letters outside $\mask|_\infty$. 
    Recall that we assumed our sets are $\mask$ oblivious, and therefore $\parikh(u)\in \lin(\vec{b},P\cup E_\mask)$. Moreover, $\parikh(z)=\mask$, since $z\in L_\mask$. In particular, we can find a permutation $u'$ of $u$ such that $u'\in \regL(\cN)$. Denote $w'=u'z$, then $w'$ satisfies the conditions above, so $w\in \ewinL(\cB)$ and we are done. \qed
\end{proof}


Next, we show that jumping languages are strictly less expressive that $\exists\boxplus$ languages.
\begin{example}
    \label{xmp:ewin_vs_JBA}
    Consider the automaton $\cA$ in \cref{fig:intersection_NBW_fails}. We claim that $\ewinL(\cA)$ is not the jumping language of any automaton. Indeed, it suffices to show that $\ewinL(\cA)$ is not permutation invariant. Trivially, $(ab)^\omega \in \ewinL(\cA)$. However, note that $(ab)^\omega \sim (a^nb^n)_{n=1}^\infty$ (i.e., the word $aba^2b^2a^3b^3\cdots $), since their Parikh images are both $(\infty,\infty)$. The latter, however, is not in $\ewinL(\cA)$, since any $k$-window permutation of $(a^nb^n)_{n=1}^\infty$ will eventually reach windows of the form $a^k$ (inside a long enough sequence $a^n$), so any permutation of it will cause $\cA$ to reject.
    \qed 
\end{example}

We finally show that $\omega$-regular languages are incomparable with $\exists \boxplus$ languages:
\begin{example}
    \label{xmp:ewin_no_NBW}
    Recall the automaton $\cA$ in \cref{fig:JBA_no_equiv_NBW}, with $\regL(\cA)=(ab)^*\cdot c^\omega$. Observe that $\ewinL(\cA)=\jL(\cA)=\{u\cdot c^\omega\mid u\in \{a,b,c\}^* \wedge \numl{u}{a}=\numl{u}{b}\}$. Indeed, if $w=u\cdot c^\omega$ as described, then $w\in \kwinL[|u|](\cA)$ by rearranging the equal number of $a$'s and $b$'s to $(ab)^*\cdot c^\omega$. Conversely, if $w\in \ewinL(\cA)$ then in particular $w\in \jL(\cA)$ (indeed, every $\exists\boxplus$-permutation is also a permutation). Thus, $\ewinL(\cA)$ is not $\omega$-regular.

    For the converse (i.e., an $\omega$-regular language that is not $\exists\boxplus$), we argue that every $\ewinL(\cA)$ is $k$-window invariant for every $k$, and since e.g., $\{(ab)^\omega\}$ is not 2-window invariant, then it is not recognizable as an $\exists\boxplus$ language.

    Indeed, consider $w\in \ewinL(\cA)$, and let $w_1\simkw w$ for some $k$. Then there exists $w_2 \simkw[k']w$ such that $w_2\in \regL(\cA)$ and $k'\in \bbN$. But then $w_2\simkw[k k'] w$ and $w_1\simkw[kk'] w$, so by transitivity $w_1\simkw[kk']w_2$, and so $w_2\in \ewinL(\cA)$. Thus, $\ewinL(\cA)$ is $k$-window invariant.
    \qed
\end{example}

\begin{remark}[Alternative Semantics]
\label{rmk:alternative_exists_semantics}
    Note that instead of the $\exists\boxplus$ semantics, we could require that there exists a partition of $w$ into windows of size \emph{at most} $k$ for some $k\in \bbN$. All our proofs carry out under this definition as well. 
\end{remark}

\section{Future Research}
\label{sec:future}
We have introduced three semantics for jumping automata over infinite words.
For the existing definitions, the class of $\exists\boxplus$ semantics is yet to be explored for closure properties and decision problems.

In a broader view, it would be interesting to consider quantitative semantics for jumping -- instead of the coarse separation between letters that occur finitely and infinitely often, one could envision a semantics that takes into account e.g., the frequency with which letters occur. 
Alternatively, one could return to the view of a ``jumping reading head'', and place constraints over the strategy used by the automaton to move the head (e.g., restrict to strategies implemented by one-counter machines).



%
%
%
\bibliographystyle{splncs04}
\bibliography{main}

\appendix

\section{Proofs}
\label{apx:proofs}

\subsection{$\omega$-regular languages as union}
\label{apx:omega_regular_union}
The following claim is a folklore simple exercise, and we bring it for completeness.
\begin{claim}
Consider an automaton $\cA=\tup{Q,\Sigma,\delta,q_0,\alpha}$, then we can write $\regL(\cA)=\bigcup_{i=1}^k S_i\cdot T_i^\omega$ where the $S_i$ and $T_i$ are regular languages.
\end{claim}
\begin{proof}
    For states $p,q\in Q$, denote $\cA_p^q$ the automaton $\cA$ with initial state $p$ and accepting states $\{q\}$. Observe that $\cA$ accepts a word if there is a run of $\cA$ that starts in $q_0$, and visits some state $q\in \alpha$ infinitely often (indeed, since $Q$ is finite, then if $\alpha$ is visited infinitely often, so is some state in $\alpha$).

    We thus have $\regL(\cA)=\bigcup_{q\in \alpha}\finL(\cA_{q_0}^q)\cdot \finL(\cA_q^q)^\omega$.
    \qed
\end{proof}

\subsection{Proof of \cref{thm:WJW to NBW}}
\label{apx:WJW to NBW}
\begin{proof}
Intuitively, we construct $\cB_k$  so that it stores in a state the Parikh image of $k$ letters, then nondeterministically simulates $\cA$ on all words with the stored Parikh image, while noting when an accepting state may have been traversed. We illustrate the construction for $k=2$ in \cref{fig:example_WJW_to_NBW}. 

We now formalize this intuition.

    Let $\cA=\tup{Q,\Sigma,\delta,q_0,\alpha}$. We construct $\cB_k=\tup{Q',\Sigma,\delta',q_0',\alpha'}$ as follows.
    \begin{itemize}
        \item Let $Q_\alpha = \{q_\alpha\ | q\in Q\}$ be a copy of $Q$, and $P = \{\parikh(x) | x\in\Sigma^{<k}\}\subseteq \{0,\ldots,k-1\}^{|\Sigma|}$ be the set of Parikh images of all words of length $<k$ (note that $P$ is finite), then the states of $\cB_k$ are $Q'=(Q\times P) \cup Q_\alpha$.
        \item The initial state is $q_0'=(q_0,\vec{0})$ (Note that $\parikh(\epsilon)=\vec{0}$).
        \item $\alpha'=Q_\alpha$.
        \item The transition function is defined as follows. For a state $(q,\vec{v})$ and letter $\sigma\in \Sigma$ let $\vec{u}=\vec{v}+\parikh(\sigma)$. If $\sum_{\sigma\in \Sigma}\vec{u}(\sigma)<k$ we define $\delta((q,\vec{v}),\sigma)=(q,\vec{u})$. 
        If $\sum_{\sigma\in \Sigma}\vec{u}(\sigma)=k$, define $\delta^*(q,\vec{u})$ to be the set of states reachable from $q$ upon reading a word $x$ with $\parikh(x)=\vec{u}$, and define $\delta_\alpha^*(q,\vec{u})$ to be the states reachable from $q$ upon reading a word $x$ with $\parikh(x)=\vec{u}$ via a run that visits $\alpha$. We then have
        \[
        \delta'((q,\vec{v}),\sigma)=\{(p,\vec{0})\mid p\in \delta^*(q,\vec{v})\}\cup \{p_\alpha\mid p\in \delta_\alpha^*(q,\vec{v})\}
        \]
        Finally, for every $q\in Q_\alpha$ we define $\delta'(q,\sigma)=\delta'((q,\vec{0}),\sigma)$ (that is, the $Q_\alpha$ copy behaves identically to the $\vec{0}$ copy of $Q$).
    \end{itemize}

    We now claim that $\kwinL(\cA)=\regL(\cB_k)$:
    Let $w\in\kwinL(\cA)$ and write $w=x_1x_2x_3...$ where $|x_i|=k$ for all $i$. Then, there is $w'=y_1y_2y_3... \simkw w$ (with all $|y_i|=k$) such that $\cA$ has an accepting run $\rho$ on $w'$. Denote by $q_0,q_1,\ldots$ the states where for $i\ge 1$ we have $q_i=\rho_{ik}$. That is, the states that $\rho$ visits after reading each $y_i$.
    Observe that for all $i$ we have that $q_{k(i+1)}\in\delta^*(q_{ik},\psi(y_{i+1}))=\delta^*(q_{ik},\psi(x_{i+1}))$. Thus, there is a run of $\cB$ that upon reading $x_{i+1}$ reaches the state $(q_{(i+1)k},\vec{0})$ if $\rho$ does not go through an accepting state on $y_{i+1}$ or $(q_{(i+1)k})_\alpha$ if it does. Since $\rho$ traverses infinitely many accepting states, so does the run of $\cB_k$ on $w$, and thus $w\in \regL(\cB_k)$.
    
    Conversely, assume $w\in\regL(\cB_k)$ and again write $w=x_1x_2x_3...$ with $|x_i|=k$. Then there is an accepting run of $\cB_k$ on $w$. By viewing the run in consecutive $k$-windows, we can construct words $y_1,y_2,\ldots$ such that $y_i\sim x_i$ for all $i$ and such that $\cA$ accepts $w'=y_1\cdot y_2\cdots$. Indeed, in every $k$-window, $\cB$ simulates the run of $\cA$ on some permutation of the read word. We conclude that $w\in \kwinL(\cA)$.
    \qed 
\end{proof}

\end{document}